\documentclass[english, 12pt]{article}

\usepackage{amsmath,amsthm,amssymb}
\usepackage{mathtools}
\usepackage{babel}

\usepackage[a4paper,top=3cm, bottom=3cm, left=2.75cm, right=2.75cm]{geometry}
\usepackage{fourier}
\usepackage[T1]{fontenc}
\usepackage[utf8]{inputenc}
\usepackage{microtype}

\usepackage{tikz}
\usetikzlibrary{arrows.meta,arrows}
\usepackage[round]{natbib}
\usepackage{authblk}
\usepackage[colorlinks=true,citecolor=blue,linkcolor=blue,urlcolor=blue]{hyperref}
\usepackage[font=small,labelfont=bf, width=.8\textwidth]{caption}

\theoremstyle{plain} 
\newtheorem{theorem}{Theorem}
\newtheorem{lemma}{Lemma}
\newcommand{\bem}[1]{\emph{\textbf{#1}\/}}
\title{\large\textbf{Von Neumann's minimax theorem \\ through Fourier-Motzkin elimination}\thanks{I thank Sergiu Hart and Bernhard von Stengel for detailed comments.}}
\author{\textsc{\normalsize Mark Voorneveld}}
\affil{\small Department of Economics, Stockholm School of Economics, Box 6501, 113 83 Stockholm, Sweden, \url{mark.voorneveld@hhs.se}}

\date{\small \today}

\begin{document}

\maketitle

\begin{abstract}
Fourier-Motzkin elimination, a standard method for solving systems of linear inequalities, leads to an elementary, short, and self-contained proof of von Neumann’s minimax theorem.


\end{abstract}

\section{Introduction}

Fourier-Motzkin elimination (FME) is a standard method to solve systems of linear inequalities. \citet{Fourier1826} described it for a specific problem and \citet{Motzkin1936} rediscovered it in his PhD dissertation; \citet{Khachiyan2009} gives an overview. Analogous to its better-known sibling Gaussian elimination for systems of linear equations, an elimination step involves generating new inequalities, implied by the original ones, from which an unknown is `eliminated' because it gets a coefficient equal to zero. But in contrast, successive rounds of FME can lead to an exponential increase in the number of inequalities \citep[p.~156]{Schrijver1986}: FME is handy for proving theorems, but computationally inefficient.

FME has been used as a stepping stone to establish intermediate results --- like some of the numerous variants of Farkas’ lemma or the duality theorem in linear programming\footnote{An excellent overview of connections between the minimax theorem, variants of Farkas' lemma, and linear programming duality is \citet{vonStengel2024}.} --- which in their turn can be used to prove the minimax theorem of \citet{vonNeumann1928}. The main result in our paper (the proof of Theorem \ref{thm: FME to minimax} in Section~\ref{sec: minimax}) is that such intermediate steps are unnecessary: FME alone easily establishes the minimax theorem (Fig.~\ref{fig: FME to minimax}(a)).

Of course, one can also use FME to prove the minimax theorem \emph{with\/} a variant of Farkas' lemma as an intermediate step. \citet{Ville1938} gave the first Farkas-style proof of the minimax theorem. When \citeauthor{vNM1944} published their \textit{Theory of Games and Economic Behavior\/} \citeyearpar{vNM1944}, they included his proof instead of the much longer one by \citet{vonNeumann1928}; \citet[Ch.~II]{Owen1968} is an early textbook treatment. Proofs along these lines often rely on topology to show that the problems where players optimize their assured payoff achieve a solution (Fig.~\ref{fig: FME to minimax}(b)). For completeness, such a proof is included in Section \ref{subsec: Farkas and topology}; Section \ref{subsec: Farkas only} then does a topology-free proof using Farkas (Fig.~\ref{fig: FME to minimax}(c)). To start off, Section \ref{sec: FME} summarizes FME and uses it to prove one variant of Farkas' lemma.

\begin{figure}[ht]
\tikzstyle{box} = [rounded corners, draw, very thick, minimum height = 0.9cm, fill = gray!30, align=center]
\tikzstyle{arrow} = [line width=1pt,-{Stealth[length=3mm, width=2mm]}]
\begin{center}
\begin{tikzpicture}
\node at (-1.2,0) {(a):};
  \node (FME1) at (0,0) [box] {FME};
  \node (minimax1) at (10,0) [box] {minimax theorem};
  \draw [arrow] (FME1) -- (minimax1);
\end{tikzpicture}
\end{center}
\begin{center}
\begin{tikzpicture}
  \node at (-1.2,0) {(b):};
  \node (FME) at (0,0) [box] {FME};
  \node (Farkas) at (4.5,0) [box] {(variants of) \\ Farkas' lemma};
  \node (minimax) at (10,0) [box] {minimax theorem};
  \node (topology) at (10,-2) [box] {topology};
  \draw [arrow] (FME) -- (Farkas);
  \draw [arrow] (Farkas) -- (minimax);
  \draw [arrow] (topology) -- (minimax);
\end{tikzpicture}
\end{center}

\begin{center}
\begin{tikzpicture}
\node at (-1.2,0) {(c):};
  \node (FME) at (0,0) [box] {FME};
  \node (Farkas) at (4.5,0) [box] {(variants of) \\ Farkas' lemma};
  \node (minimax) at (10,0) [box] {minimax theorem};
  \draw [arrow] (FME) -- (Farkas);
  \draw [arrow] (Farkas) -- (minimax);
\end{tikzpicture}
\end{center}

\caption{Fourier-Motzkin elimination, Farkas' lemma, and the minimax theorem: three routes.}\label{fig: FME to minimax}
\end{figure}

\medskip
\noindent \textsc{notation:} $m$ and $n$ denote positive integers and $[n] = \{1, \ldots, n\}$. Vectors $\mathbf{0}$ and $\mathbf{1}$ have all components equal to zero and one, respectively; $e_i$ is the standard basis vector with $i$-th component equal to one and all others equal to zero. The identity matrix is $I$ and $O$ is a zero matrix. The dimension of such vectors and matrices depends on the context. Unless explicitly written otherwise, vectors are column vectors. We use superscript $\top$ for transposes. For vectors of equal dimension, $v \leq w$ means that $v_i \leq w_i$ for all components $i$. The number of elements in a finite set $S$ is $|S|$.

\section{From Fourier-Motzkin elimination to Farkas' lemma}\label{sec: FME}

Given matrix $A \in \mathbb{R}^{m \times n}$ and vector $b \in \mathbb{R}^m$, consider the system $Ax \leq b$ of $m$ linear inequalities in $n$ unknowns $x = (x_1, \ldots, x_n)$:
\[
a_{i1} x_1 + \cdots + a_{in} x_n \leq b_i \qquad \qquad i \in [m].
\]
Divide the $m$ inequalities into three types, depending on whether the unknown $x_1$ has a coefficient (G)reater than, (E)qual to, or (L)ess than zero:
\[
G = \bigl\{ i \in [m]: a_{i1} > 0\bigr\}, \qquad E = \bigl\{ i \in [m]: a_{i1} = 0\bigr\}, \qquad L = \bigl\{ i \in [m]: a_{i1} < 0\bigr\}.
\]
Inequalities $i \in G$ give \emph{upper\/} bounds on $x_1$. Multiply both sides by $1/a_{i1} > 0$ and rewrite:
\begin{equation}\label{eq: FMpos}
\text{for all } i \in G: \qquad x_1 \leq \frac{1}{a_{i1}} b_i - \frac{1}{a_{i1}} \left(a_{i2} x_2 + \cdots + a_{in} x_n \right).
\end{equation}
Inequalities $j \in L$ give \emph{lower\/} bounds on $x_1$. Multiply both sides by $-1/a_{j1} > 0$ and rewrite:
\begin{equation}\label{eq: FMneg}
\text{for all } j \in L: \qquad \frac{1}{a_{j1}} b_j - \frac{1}{a_{j1}} \left(a_{j2} x_2 + \cdots + a_{jn} x_n \right) \leq x_1.
\end{equation}
Inequalities $k \in E$ impose \emph{no\/} bounds on $x_1$, since its coefficient is zero:
\begin{equation}\label{eq: FMzero}
\text{for all } k \in E: \qquad 0 x_1 + a_{k2} x_2 + \cdots + a_{kn} x_n \leq b_k.
\end{equation}
So there is a solution $x = (x_1, \ldots, x_n)$ to linear inequalities $Ax \leq b$ if and only if $(x_2, \ldots, x_n)$, the vector of unknowns from which $x_1$ was eliminated, satisfies the $|E|$ linear inequalities in \eqref{eq: FMzero} and the $|G| \cdot |L|$ linear inequalities that say that each upper bound in \eqref{eq: FMpos} is indeed greater than or equal to each lower bound in \eqref{eq: FMneg}, leaving room to squeeze the remaining unknown $x_1$ in between: for all $i \in G$ and all $j \in L$,
\[ \frac{1}{a_{j1}} b_j - \frac{1}{a_{j1}} \left(a_{j2} x_2 + \cdots + a_{jn} x_n \right) \leq \frac{1}{a_{i1}} b_i - \frac{1}{a_{i1}} \left(a_{i2} x_2 + \cdots + a_{in} x_n \right),
\]
or, equivalently, after rearranging and noticing that the $x_1$-terms cancel out,
\begin{equation}\label{eq: FME nonneg}
\frac{1}{a_{i1}} \left(a_{i1} x_1 + \cdots + a_{in} x_n \right) - \frac{1}{a_{j1}} \left(a_{j1} x_1 + \cdots + a_{jn} x_n \right) \leq \frac{1}{a_{i1}} b_i - \frac{1}{a_{j1}} b_j.
\end{equation}
These inequalities are nonnegative combinations of those in the original system $Ax \leq b$, i.e., of the form $y^{\top} Ax \leq y^{\top} b$ for some nonnegative vector $y = (y_1, \ldots, y_m)$: for inequality $k \in E$ in \eqref{eq: FMzero}, $y = e_k$; for inequality \eqref{eq: FME nonneg} arising from $i \in G$ and $j \in L$, $y = \tfrac{1}{a_{i1}} e_i + \left(- \tfrac{1}{a_{j1}} \right) e_j$.

Let $Y_1$ be a matrix with these $y$'s as its rows. Then we can summarize the result of eliminating the unknown $x_1$ as follows. There is a nonnegative matrix $Y_1$ such that the $|E| + |G| \cdot |L|$ conditions in \eqref{eq: FMzero} and \eqref{eq: FME nonneg} can be written as $Y_1 A x \leq Y_1 b$. This system `eliminates' $x_1$ (coefficient matrix $Y_1 A$ has only zeroes in its first column) and therefore imposes conditions on $(x_2, \ldots, x_n)$ only.\footnote{If there are only upper bounds on $x_1$ (i.e., $|E| = |L| = 0$), we find $|E| + |G| \cdot |L| = 0$ restrictions on $(x_2, \ldots, x_n)$: each such vector can be augmented with an $x_1$ to find a solution of $Ax \leq b$ by choosing $x_1$ less than or equal to the smallest of these finitely many upper bounds. In that case, we can take $Y_1$ to be the zero matrix. A similar comment applies if there are only lower bounds on $x_1$.} Its solutions are those $(x_2, \ldots, x_n)$ for which we can find a suitable $x_1$ such that $x = (x_1, \ldots, x_n)$ solves the original system $Ax \leq b$.

Iterating to eliminate further unknowns, we find nonnegative matrices $Y_1, \ldots, Y_n$ such that after $k \in \{1, \ldots, n-1\}$ rounds of elimination the system of linear inequalities
\[
YA x \leq Yb \qquad \text{(with nonnegative matrix $Y = Y_k Y_{k-1} \cdots Y_1$)}
\]
eliminates $x_1, \ldots, x_k$ (coefficient matrix $YA$ has only zeroes in its first $k$ columns) and therefore imposes conditions on $x_{k+1}, \ldots, x_n$ only. Moreover, its solutions are
\[
\bigl\{(x_{k+1}, \ldots, x_n): \text{there exist } (x_1, \ldots, x_k) \text{ such that } x = (x_1, \ldots, x_n) \text{ solves } Ax \leq b\bigr\}.
\]
Multiplying an inequality by a positive scalar does not affect its solutions, so we may assume, whenever convenient, that the entries in a specific nonzero column of coefficient matrix $YA$ are either $-1$, $0$, or $1$.

Eliminating \emph{all\/} unknowns gives inequalities $YA x \leq Yb$ with nonnegative matrix $Y = Y_n Y_{n-1} \cdots Y_1$ such that \emph{all\/} entries of coefficient matrix $YA$ are zero: $YA$ is the zero matrix $O$. If $x$ solves the original system $Ax \leq b$, left multiplication with $Y$ gives $\mathbf{0} = O x = YA x \leq Yb$, i.e., $Yb \geq \mathbf{0}$. If, in contrast, $Yb$ has a negative coordinate $(Yb)_i < 0$, let $y^{\top}$ be the nonnegative $i$-th row of $Y$. Then $y^{\top} A = \mathbf{0}^{\top}$ (the $i$-th row of the zero matrix $YA$) and $y^{\top} b = (Yb)_i < 0$. Hence, there is an $x$ solving $Ax \leq b$ or a $y$ solving $y^{\top} A = \mathbf{0}^{\top}$, $y^{\top} b < 0$, $y \geq \mathbf{0}$. We can't have both: such $x$ and $y$ would imply
\[
0 = \mathbf{0}^{\top} x = (y^{\top} A) x = y^{\top} (Ax) \leq y^{\top} b < 0,
\]
which is impossible. This proves one of the many variants of Farkas' lemma (Theorem 1 in \citet{Fan1957}, Theorem 2.7 in \citet{Gale1960}):

\begin{theorem}[Farkas' lemma]\label{thm: Farkas FME}
Given a matrix $A \in \mathbb{R}^{m \times n}$ and a vector $b \in \mathbb{R}^m$, exactly one of the following two statements is true:
\begin{enumerate}
\item There is an $x \in \mathbb{R}^n$ solving $Ax \leq b$;
\item There is a $y \in \mathbb{R}^m$ solving $y^{\top} A = \mathbf{0}^{\top}$, $y^{\top} b < 0$, $y \geq \mathbf{0}$.
\end{enumerate}
\end{theorem}

\section{Von Neumann's minimax theorem}\label{sec: minimax}

A matrix game is defined by a matrix $A \in \mathbb{R}^{m \times n}$ where entry $a_{ij}$ is the payoff \emph{to\/} player 1 \emph{from\/} player 2 if the former chooses row $i \in [m]$ and the latter column $j \in [n]$. If player 1 assigns probabilities $p = (p_i)_{i \in [m]}$ to the $m$ rows ($p \geq \mathbf{0}$ and $\mathbf{1}^{\top} p = 1$) and player 2 assigns probabilities $q = (q_j)_{j \in [n]}$ to the $n$ columns ($q \geq \mathbf{0}$ and $\mathbf{1}^{\top} q = 1$), player 1's expected payoff is $p^{\top} A q$. Given $p$, this payoff is a convex combination of those in row vector $p^{\top} A$: player 1 is sure that the expected payoff is at least $v$ for any number $v$ with $p^{\top} A \geq v \mathbf{1}^{\top}$. Through a clever choice of $p$, the highest expected payoff this player can assure therefore solves
\begin{align}
& \underset{p,v}{\text{maximize }} v \quad \text{subject to} \quad p^{\top} A \geq v \mathbf{1}^{\top}, \enspace p \geq \mathbf{0}, \enspace \mathbf{1}^{\top} p= 1. \label{eq: pl 1}
\intertext{Similarly, player 2, who wants to pay as little as possible, can make sure that player 1's expected payoff is not more than the solution to}
& \underset{q,v}{\text{minimize }} v \quad \text{subject to} \quad A q \leq v \mathbf{1}, \enspace q \geq \mathbf{0}, \enspace \mathbf{1}^{\top} q = 1. \label{eq: pl 2}
\end{align}
\begin{theorem}[von Neumann's minimax theorem]
In each matrix game $A$, the players' optimization problems \eqref{eq: pl 1} and \eqref{eq: pl 2} have a solution and their optimal values are the same.
\end{theorem}
For a proof, it suffices to find $(p^*, q^*, v^*) \in \mathbb{R}^m \times \mathbb{R}^n \times \mathbb{R}$ such that $(p^*,v^*)$ is feasible in \eqref{eq: pl 1} and $(q^*,v^*)$ is feasible in \eqref{eq: pl 2}. Then $(p^*,v^*)$ solves \eqref{eq: pl 1} since each other feasible $(p,v)$ has
\[
v = v \left(\mathbf{1}^{\top} q^* \right) = (v \mathbf{1}^{\top}) q^* \leq (p^{\top} A) q^* = p^{\top} (A q^*) \leq p^{\top} (v^* \mathbf{1}) = v^*.
\]
Likewise, $(q^*,v^*)$ solves \eqref{eq: pl 2} and both problems have optimal value $v^*$. So we rephrase:

\begin{theorem}[Minimax rephrased]\label{thm: FME to minimax}
In each matrix game $A$, there are $(p^*, q^*, v^*) \in \mathbb{R}^m \times \mathbb{R}^n \times \mathbb{R}$ such that $(p^*, v^*)$ is feasible in \eqref{eq: pl 1} and $(q^*, v^*)$ is feasible in \eqref{eq: pl 2}.
\end{theorem}

Thus, the minimax theorem is a task for FME: finding solutions to linear inequalities!

\begin{proof}[Proof of Theorem 3]
Write the constraints in \eqref{eq: pl 2} as linear inequalities $Aq - v \mathbf{1} \leq \mathbf{0}$, $- I q \leq \mathbf{0}$, $\mathbf{1}^{\top} q \leq 1$, and $(- \mathbf{1})^{\top} q \leq -1$. In matrix form:
\[
\widehat{A} \begin{bmatrix} q \\ v \end{bmatrix} \leq \widehat{b} \qquad \text{ with } \qquad
\widehat{A} = \begin{bmatrix*}[r]
A\phantom{^{\top}} & - \mathbf{1} \\ -I\phantom{^{\top}} & \mathbf{0} \\ \mathbf{1}^{\top} & 0 \\ - \mathbf{1}^{\top} & 0
\end{bmatrix*},
\quad
\widehat{b} = \begin{bmatrix*}[r]
\mathbf{0} \\ \mathbf{0} \\ 1 \\ -1
\end{bmatrix*}.
\]
Since these inequalities have solutions and in each solution $(q,v)$, the value of $v$ is bounded from below by, e.g., $A$'s lowest payoff, the same is true if we eliminate $q_1, \ldots, q_n$ by FME. By construction (recall Section 2), this yields a system of linear inequalities with solutions
\begin{equation}\label{eq: set v}
\left\{v: \text{ there is a $q$ such that } \widehat{A} \begin{bsmallmatrix} q \\ v \end{bsmallmatrix} \leq \widehat{b}\right\},
\end{equation}
and which can be expressed in the form
\begin{equation}\label{eq: only v}
Y \widehat{A} \begin{bmatrix} q \\ v \end{bmatrix} \leq Y \widehat{b}
\end{equation}
for a nonnegative matrix $Y$ such that in the coefficient matrix $Y \widehat{A}$ all $q_i$ variables have coefficient 0 and those of $v$ lie in $\{-1,0,1\}$.\footnote{We won't use this, but we can exclude coefficient $1$ on $v$. Recall that in FME, positive coefficients yield upper bounds on its variable. But $v$ has no upper bound: if $(q,v)$ is feasible, then so is $(q,v')$ for all $v' > v$.} As $v$ is bounded from below and, in FME, lower bounds on a variable come from inequalities where it has a negative coefficient, \eqref{eq: only v} includes inequalities where $v$ has coefficient $-1$: inequalities of the form
\[
0 q_1 + \cdots + 0 q_n - v \leq - L \qquad \text{(equivalently, $v \geq L$)}
\]
where $L$ is some lower bound. Let $v^*$ be the largest of the finitely many lower bounds. So
\[
0 q_1 + \cdots + 0 q_n - v = [\mathbf{0}^{\top} \, -1] \begin{bmatrix} q \\ v \end{bmatrix} \leq - v^*
\]
is an inequality in \eqref{eq: only v}: there is a row $y^{\top}$ of $Y$ with $y^{\top} \widehat{A} = [\mathbf{0}^{\top} \, -1]$ and $y^{\top} \widehat{b} = -v^*$, i.e., a nonnegative vector $y = (p^*, s, \delta, \varepsilon) \in \mathbb{R}^m \times \mathbb{R}^n \times \mathbb{R} \times \mathbb{R}$ with
\[
(p^*)^{\top} A - s^{\top} I + (\delta - \varepsilon) \mathbf{1}^{\top} = \mathbf{0}^{\top}, \qquad (p^*)^{\top} (- \mathbf{1}) = -1, \qquad \delta - \varepsilon = - v^*.
\]
Consequently, $p^* \geq \mathbf{0}$, $\mathbf{1}^{\top} p^* = 1$, and by nonnegativity of $s$,
\[
(p^*)^{\top} A = s^{\top} I + (\varepsilon - \delta) \mathbf{1}^{\top} \geq \mathbf{0}^{\top} + v^* \mathbf{1}^{\top} = v^* \mathbf{1}^{\top},
\]
making $(p^*, v^*)$ feasible in \eqref{eq: pl 1}. Finally, $v^*$ lies in \eqref{eq: set v}: as the largest of the lower bounds on $v$, it is $v$'s smallest possible value. So $(q^*, v^*)$ feasible in \eqref{eq: pl 2} for some $q^*$.
\end{proof}

This proof uses no topological arguments, like compactness of strategy spaces or continuity of payoffs. And it applies equally well if payoffs and probabilities lie in other ordered fields than the reals, like the set of rational numbers: FME works in any ordered field. The same comment applies to our proof in Section \ref{subsec: Farkas only}. The duality theorem of linear programming can be proved in essentially the same way \citep[Thm.~2.29]{KippMartin1999}.

\section{Proofs of the minimax theorem using Farkas' lemma}\label{sec: Farkas}

Our main result is in the previous section: FME directly implies the minimax theorem. To complete the overview in Figure~\ref{fig: FME to minimax}, Section~\ref{subsec: Farkas and topology} gives a more traditional proof using the particular version of Farkas' lemma we derived from FME in Theorem \ref{thm: Farkas FME}, plus a little topology, while Section~\ref{subsec: Farkas only} shows that the latter appeal to topology can be dispensed with.

Say that player 1 can \bem{assure} (a payoff of at least) $v \in \mathbb{R}$ if there is a $p$ such that $(p,v)$ is feasible in \eqref{eq: pl 1} and that player 2 can \bem{assure} (to pay at most) $w \in \mathbb{R}$ if there is a $q$ such that $(q,w)$ is feasible in \eqref{eq: pl 2}. If both are true, i.e.,
\begin{equation}\label{eq: low high}
\text{if player 1 can assure $v \in \mathbb{R}$ and player 2 can assure $w \in \mathbb{R}$, then $v \leq w$.}
\end{equation}
Indeed, with corresponding $p$ and $q$ we find that
\[
v = v(\mathbf{1}^{\top} q) = (v \mathbf{1}^{\top}) q \leq (p^{\top} A) q = p^{\top} (Aq) \leq p^{\top} (w \mathbf{1}) = w (p^{\top} \mathbf{1}) = w.
\]

\subsection{A proof using Farkas and a topological argument}\label{subsec: Farkas and topology}

The first step relies on topology to argue that the players' optimization problems \eqref{eq: pl 1} and \eqref{eq: pl 2} have solutions. Conditions $p^{\top} A \geq v \mathbf{1}^{\top}$ in player 1's problem \eqref{eq: pl 1} say that $\min_{j \in [n]} p^{\top} A e_j \geq v$. Since we maximize $v$, the latter holds with equality in an optimum. So the problem reduces to maximizing continuous function $p \mapsto \min_{j \in [n]} p^{\top} A e_j$ over the compact set of probability vectors $p$: an optimum $(p^*, \underline{v})$ exists. Likewise, an optimum $(q^*, \overline{v})$ exists in \eqref{eq: pl 2}.

In the second step, Farkas makes its entrance. Since $(p^*,\underline{v})$ and $(q^*, \overline{v})$ are feasible in the players' optimization problems \eqref{eq: pl 1} and \eqref{eq: pl 2}, respectively, we know from \eqref{eq: low high} that $\underline{v} \leq \overline{v}$. It remains to show that they are equal. If they aren't, pick any $v$ with $\underline{v} < v < \overline{v}$. Neither player can assure $v$: by definition of the optimal values $\underline{v}$ and $\overline{v}$, player 1 cannot assure more than $\underline{v}$ and player 2 cannot assure less than $\overline{v}$. This contradicts:

\begin{lemma}
For each matrix game $A$ and scalar $v$, at least one player can assure $v$.
\end{lemma}
\begin{proof}
If player 2 can assure $v$, we are done. So assume this is not the case: there is no $q$ with $A q \leq v \mathbf{1}$, $q \geq \mathbf{0}$, and $\mathbf{1}^{\top} q = 1$. In matrix form, there is there is no solution $q$ to
\[
\begin{bmatrix*}[r]
A\phantom{^{\top}}\\ -I\phantom{^{\top}} \\ \mathbf{1}^{\top} \\ - \mathbf{1}^{\top}
\end{bmatrix*} q \leq \begin{bmatrix*}[r]
v \mathbf{1} \\ \mathbf{0} \\ 1 \\ -1
\end{bmatrix*}.
\]
By Farkas' lemma (Thm.~\ref{thm: Farkas FME}) there is a nonnegative vector $(p, s, \delta, \varepsilon) \in \mathbb{R}^m \times \mathbb{R}^n \times \mathbb{R} \times \mathbb{R}$ with
\begin{equation}\label{eq: first Farkas}
(i): \enskip p^{\top} A - s^{\top} I + (\delta - \varepsilon) \mathbf{1}^{\top} = \mathbf{0}^{\top} \qquad \text{and} \qquad
(ii): \enspace v (p^{\top} \mathbf{1}) + (\delta - \varepsilon) < 0.
\end{equation}
If $p = \mathbf{0}$, part $(i)$ and $s \geq \mathbf{0}$ give $(\delta - \varepsilon) \mathbf{1}^{\top} = s^{\top} I \geq \mathbf{0}^{\top}$, i.e., $\delta - \varepsilon \geq 0$. But part $(ii)$ says the opposite. So $p$ is nonnegative and $p \neq \mathbf{0}$. If we divide all entries of $(p, s, \delta, \varepsilon)$ by $\mathbf{1}^{\top} p > 0$, the nonnegativity condition and the (in)equalities in \eqref{eq: first Farkas} remain true, so we may assume without loss of generality that $p$'s coordinates sum to 1: $\mathbf{1}^{\top} p =1$. Plug this into \eqref{eq: first Farkas}:
\[
\varepsilon - \delta > v \qquad \text{and} \qquad p^{\top} A = s^{\top} I + (\varepsilon - \delta) \mathbf{1}^{\top} \geq \mathbf{0} + (\varepsilon - \delta) \mathbf{1}^{\top} \geq v \mathbf{1}^{\top}.
\]
We conclude that $p^{\top} A \geq v \mathbf{1}^{\top}$, $p \geq \mathbf{0}$, and $\mathbf{1}^{\top} p = 1$: player 1 can assure $v$ by playing $p$.
\end{proof}

\subsection{A proof using Farkas only}\label{subsec: Farkas only}

In matrix form, the desired $(p^*, q^*, v^*)$ in Theorem \ref{thm: FME to minimax} must solve
\[
\begin{bmatrix*}[r]
\phantom{-}O\phantom{^{\top}} & \phantom{-}A\phantom{^{\top}} & - \mathbf{1} \\
\phantom{-}O\phantom{^{\top}} & -I\phantom{^{\top}} & \mathbf{0} \\
\phantom{-}\mathbf{0}^{\top} & \phantom{-}\mathbf{1}^{\top} & 0 \\
\phantom{-}\mathbf{0}^{\top} & -\mathbf{1}^{\top} & 0 \\
-A^{\top} & \phantom{-}O\phantom{^{\top}} & \mathbf{1} \\
-I\phantom{^{\top}} & \phantom{-}O\phantom{^{\top}} & \mathbf{0} \\
\phantom{-}\mathbf{1}^{\top} & \phantom{-}\mathbf{0}^{\top} & 0 \\
-\mathbf{1}^{\top} & \phantom{-}\mathbf{0}^{\top} & 0
\end{bmatrix*}
\begin{bmatrix} p \\ q \\ v \end{bmatrix} \leq
\begin{bmatrix*}[r]
\mathbf{0} \\ \mathbf{0} \\ 1 \\ -1 \\ \mathbf{0} \\ \mathbf{0} \\ 1 \\ -1
\end{bmatrix*}.
\]
If no such solution exists, Farkas' lemma (Theorem \ref{thm: Farkas FME}) says there is a nonnegative vector
\begin{equation}\label{eq: long}
(x, s_x, \delta_x, \varepsilon_x, y, s_y, \delta_y, \varepsilon_y) \in \mathbb{R}^m \times \mathbb{R}^n \times \mathbb{R} \times \mathbb{R} \times \mathbb{R}^n \times \mathbb{R}^m \times \mathbb{R} \times \mathbb{R}
\end{equation}
such that
\begin{align}
- y^{\top} A^{\top} - s_y^{\top} I + (\delta_y - \varepsilon_y) \mathbf{1}^{\top} & = \mathbf{0}^{\top}, \label{eq: L1} \\
x^{\top} A - s_x^{\top} I + (\delta_x - \varepsilon_x) \mathbf{1}^{\top} & = \mathbf{0}^{\top}, \label{eq: L2} \\
- x^{\top} \mathbf{1} + y^{\top} \mathbf{1} & = 0, \label{eq: L3} \\
(\delta_x - \varepsilon_x) + (\delta_y - \varepsilon_y) & < 0 \label{eq: L4}
\end{align}
Recall that $x$ and $y$ are nonnegative. By equation \eqref{eq: L3}, if one of them is $\mathbf{0}$, then so is the other. Using nonnegativity of $s_x$ and $s_y$, substitution in \eqref{eq: L1} and \eqref{eq: L2} would then imply that $\delta_y - \varepsilon_y \geq 0$ and $\delta_x - \varepsilon_x \geq 0$, contradicting \eqref{eq: L4}. So $x \geq \mathbf{0}$ and $y \geq \mathbf{0}$ are distinct from $\mathbf{0}$. By \eqref{eq: L3}, the sum of their coordinates is the same, say $c > 0$. If we rescale the vector in \eqref{eq: long} by $1/c$, the nonnegativity conditions and all conditions \eqref{eq: L1} to \eqref{eq: L4} remain true. Thus, without loss of generality, we may assume that the coordinates of $x$ and $y$ both sum to one. Combined with \eqref{eq: L2} we conclude that
\begin{align*}
x^{\top} A & \geq (\varepsilon_x - \delta_x) \mathbf{1}^{\top}, & x & \geq \mathbf{0}, & \mathbf{1}^{\top} x & = 1,
\intertext{making $(x, \varepsilon_x - \delta_x)$ feasible in \eqref{eq: pl 1}. Likewise, with \eqref{eq: L1} and taking transposes,}
Ay & \leq (\delta_y - \varepsilon_y) \mathbf{1}, & y & \geq \mathbf{0}, & \mathbf{1}^{\top} y & = 1,
\end{align*}
making $(y, \delta_y - \varepsilon_y)$ feasible in \eqref{eq: pl 2}. By \eqref{eq: low high}: $\varepsilon_x - \delta_x \leq \delta_y - \varepsilon_y$, contradicting \eqref{eq: L4}. This contradiction proves the theorem.

\bibliographystyle{abbrvnat}
\bibliography{FMEminimax}

\end{document}